\documentclass[11pt]{article}

\date{}

\usepackage{amsmath,amssymb}
\usepackage{epsfig, color}
\usepackage{color,graphicx}
\usepackage{amsfonts}
\usepackage{amsmath}
\usepackage{multirow}
\usepackage{slashbox}
\usepackage{natbib}
\usepackage{appendix}

\newtheorem{example}{\noindent Example}[section]

\newtheorem{theorem}{Theorem}[section]

\newtheorem{algorithm}{Algorithm}[section]

\begin{document}

\title{A robust algorithm and convergence analysis for static replications of nonlinear payoffs}

\vskip 4mm

\author{Jingtang Ma\thanks{School of Economic Mathematics,
Southwestern University of Finance and Economics, Chengdu,
611130, P.R. China (Email: mjt@swufe.edu.cn). The work
was supported by Program for New Century Excellent Talents in University
(Grant No. NCET-12-0922).}, Dongya Deng\thanks{School of Finance,
Southwestern University of Finance and Economics, Chengdu,
611130, P.R. China (Email: 112020204015@2012.swufe.edu.cn)} and Harry Zheng
\thanks{Department of Mathematics, Imperial College, London SW7 2BZ, UK.
(Email: h.zheng@imperial.ac.uk). Corresponding author.}}

\maketitle

\begin{abstract}
In this paper we propose a new robust algorithm to find the optimal  static replicating portfolios for general nonlinear payoff functions and give the estimate of the rate of  convergence that is absent in the literature.
We choose the static replication by minimizing the error bound between the nonlinear payoff function and the linear spline approximation and derive the equidistribution equation for selecting the optimal strike prices. The numerical tests for variance swaps and swaptions and also for the static quadratic replication and the model with  counterparty risk show that the proposed algorithm is simple, fast and accurate.
 The paper  has  generalized and improved the results of the static replication and approximation in  the literature.

\end{abstract}

\bigskip


\bigskip


\section{Introduction}

It is well known that hedging a derivative is in general  much more difficult than pricing the derivative as hedging requires the determination of the feasible trading strategy whereas pricing only involves the computation of the  expected payoff which may be found with the numerical integration   or  simulation. Dynamic replication can be used for hedging with the help of the martingale representation theorem if the market is complete,  however, it is often difficult to implement as the market is in fact incomplete. Static replication is a viable alternative.

The idea of using a portfolio of options to replicate complex payoffs can date back to
\cite{Ross1976} and \cite{Breeden1978}. If options with strikes from zero to infinity are all available, then any payoff function at maturity can be replicated exactly with static hedging. Under the assumption of no arbitrage, the price of a derivative being replicated is then the total premium of the replicating options. Compared to the
dynamic replication which may incur prohibitively high transaction costs, the static replication has many advantages,
 see e.g., \cite{Derman1995}, \cite{Carr1998}, \cite{Demeterfi1999}.
Static replication with a portfolio of European calls and puts is easy to implement and does not incur running transaction costs.
\cite{Carr2013} discuss and compare
the static hedging with the delta hedging  when the underlying asset price is exposed to the possibility of jumps of random sizes  and conclude that the static hedging strongly outperforms the delta hedging.

To find a static replication one needs first to have a good approximation to the payoff function. The linear spline  approximation  is a  simple yet effective method. The key benefit of using a linear spline is that the resulting static replicating portfolio consists of simple European calls, puts, and digital options, and the weights of these options can be easily computed no matter how complex the payoff function is. In theory the approximation error can be made arbitrarily small if the number of grid points is sufficiently large  and the maximal distance of adjacent grid points is sufficiently small. In practice one has to strike the balance between the accuracy and the cost, which means one needs to choose grid points carefully to minimize the error if the number of grid points is fixed.

\cite{Demeterfi1999} use European calls and puts  with equally-spaced strike prices to replicate the log payoff, which is not optimal due to the use of the equally-spaced strike prices. \cite{Broadie2008} propose a simulation method to obtain the optimal approximation of a static replication, which minimizes the approximation error but is  computationally expensive.
\cite{Liu2010} discusses three optimal approximations of nonlinear payoffs. The first two approaches are to minimize the expected area (simple average and weighted average) enclosed by the payoff curve and the chords, which  implicitly assume the payoff function is  convex (or concave) and cannot be applied to general payoffs. The third is to minimize the expected sum of squared differences of the payoff and the replicating portfolio, which
 is computationally expensive in solving the optimization equation for complex nonlinear payoffs. The aforementioned papers do not discuss the convergence theory for the approximation.

In this paper, motivated by the idea from \cite{deBoor1973}, we propose a new robust algorithm to find the optimal approximation for  general nonlinear payoffs and provide the convergence theory for the algorithm. We first give an estimate of the error bound between the nonlinear payoff function and the linear spline approximation. We then choose the strike prices of the static replication by minimizing the error bound  instead of the error itself. The reason of doing this is that we can derive a tractable equidistribution equation for selecting the optimal strike prices, which would be difficult if the error is to be minimized  directly.  This approach of static replication works fine if  the options with strike prices  from zero to infinity are all available.  In practice, we use options with given strike prices traded in the market and we may have no choice on the strike prices. In that case we use the static quadratic hedging to find the optimal weights of the options. The robust algorithm is again useful in computing the optimal weights with some modified payoff functions.

The main contribution of the paper is that a  robust (simple, fast and accurate) iterative algorithm is proposed to find the optimal  static replicating portfolio for general nonlinear payoff functions and the  convergence theory  is  proved. The results of the paper have improved and generalized those of  \cite{Liu2010} and others in the literature.

The paper is organized as follows. In section~2 we discuss the approximation of a nonlinear payoff function by a linear spline and a portfolio of calls and puts  and estimate the error bound (Theorem \ref{thm-error}). In section~3 we propose a robust iterative  algorithm to find the optimal strike prices in static replication  and gives the estimate of the convergence rate (Theorem \ref{thm-convergence-rate}). We also apply the  quadratic hedging to find the optimal weights of the static replicating portfolio when the number of traded calls and puts in the market are fixed and finite.
In section~4 we perform some numerical tests and compare the results with those from analytic formulas or simulations for different payoffs and asset price distributions, including the case of counterparty risk. In section~5 we conclude. In appendix we give the proofs of Theorems \ref{thm-error} and \ref{thm-convergence-rate} and the derivation of the distribution for asset price with counterparty risk.

\section{Static replication  and error bound on approximation}

In this section, we give a  formula as in \cite{Liu2010} for replicating the nonlinear payoff with a basket of European options and derive the error bound on the approximation.


Let $S$ be a nonnegative random variable, representing the asset price at maturity, and $f(S)$ the derivative value with  $f$ a continuous payoff function defined on the positive real line.
Let  $[0,+\infty)$ be partitioned by $X_0,\, X_1,\,\ldots,X_{n}$, with $0<X_0<X_1<\cdots<X_n<+\infty$ and $X_0,\, X_n$ being fixed. Then $f$ can be approximated by the following piecewise linear functions:
\begin{equation}\label{piecewise-linear}
 L_{i}(S)=\frac{X_{i+1}-S}{h_{i}}f(X_{i})+\frac{S-X_{i}}{h_{i}}f(X_{i+1}),\quad S\in [X_{i},X_{i+1}],
\end{equation}
where $h_{i}\equiv X_{i+1}-X_{i},\; i=0,1,\ldots,n-1$.
The payoff curve between $X_0$ and $X_n$ can be represented approximately by the following formula (see \cite{Liu2010}):
\begin{eqnarray}\label{replication-formula}
f(S)&\approx& \sum_{i=0}^{n-1}L_{i}(S)\textbf{1}_{X_{i}<S<X_{i+1}}\\
&=& L_{k}(X_{k})-L_{0}(X_0)\textbf{1}_{S<X_0}-L_{n-1}(X_{n})\textbf{1}_{S>X_{n}}
+b_{0}(X_0-S)\textbf{1}_{S<X_0}\nonumber\\
&+&\sum_{i=1}^{k-1}(b_{i}-b_{i-1})(X_{i}-S)\textbf{1}_{S<X_{i}} -b_{k-1}(X_{k}-S)\textbf{1}_{S<X_{k}}+  b_{k}(S-X_{k})\textbf{1}_{S>X_{k}}\nonumber\\ &+&\sum_{i=k+1}^{n-1}(b_{i}-b_{i-1})(S-X_{i})\textbf{1}_{S>X_{i}}-b_{n-1}(S-X_{n})
\textbf{1}_{S>X_{n}},\nonumber
\end{eqnarray}
where $\textbf{1}_{A}$ denotes the indicator function ($\textbf{1}_{A}=1$ for $x\in A$, and $\textbf{1}_{A}=0$ otherwise.) and $b_{i}= (f(X_{i+1})-f(X_{i}))/h_{i}$. In equation (\ref{replication-formula}), the first term of the last equality is a cash amount, the second term a cash-or-nothing put, the third term a cash-or-nothing call, the next three terms are a portfolio of European puts with strike prices $X_{i}$, $i=0,1,\ldots,k$, and the last three terms are a portfolio of European calls with strike prices $X_{i}$, $i=k+1,\ldots,n$. $X_{k}$.  Both puts and calls are likely to be out-of-money options.

When the strike prices $X_0$ and $X_{n}$ are very small and very large, respectively, the probabilities  $P(S<X_{0})$ and $P(S>X_{n})$ are extremely small,
the second, third, fourth and the last terms in the formula (\ref{replication-formula}) have little impact to the valuation and thus may be removed. In such a case, the following formula from \cite{Demeterfi1999} can be used for static replication:
\begin{eqnarray}\label{replication-formula-Deme}
f(S)&\approx& L_{k}(X_{k})+\sum_{i=1}^{k-1}(b_{i}-b_{i-1})(X_{i}-S)^+ -b_{k-1}(X_{k}-S)^+\nonumber\\
&+&  b_{k}(S-X_{k})^+ +\sum_{i=k+1}^{n-1}(b_{i}-b_{i-1})(S-X_{i})^+.
\end{eqnarray}
\cite{Liu2010} chooses the strike prices $X_0,X_1,\ldots,X_{n}$ such that the total area enclosed by the payoff curve and the chords
\[
 \sum_{i=0}^{n-1}\int_{X_{i}}^{X_{i+1}}[L_{i}(S)-f(S)]\,dS
\]
is minimized. \cite{Liu2010} also finds that the performance can be improved if  the weighted total area
\[
\sum_{i=0}^{n-1}\int_{X_{i}}^{X_{i+1}}[L_{i}(S)-f(S)]g(S)\,dS,
\]
is used, where $g$ is the density function of $S$ (conditional on today's price of the underlying).
It is clear that $f$ needs to be convex to ensure all integrands are nonnegative. For a general  payoff function $f$ we measure the error with the weighted  squared norm
 \[
  \hbox{Error}\equiv \sqrt{\sum_{i=0}^{n-1}\int_{X_{i}}^{X_{i+1}}[L_{i}(S)-f(S)]^{2}g(S)\,dS}.
 \]


The error bound on the linear spline approximation (\ref{piecewise-linear}) is given by the following theorem.
\begin{theorem}\label{thm-error}
 Assume that $f$ is continuous on $[X_0,X_n]$ and twice continuously differentiable on $(X_{i},X_{i+1})$, $ i=0,1,\ldots, n-1$, with finite second order left and right directional derivatives at $X_i$, $i=0,1,\ldots,n$.   Then
 the  error of the linear spline approximation (\ref{piecewise-linear}) to the nonlinear payoff function $f$  is bounded by
 \[
  \sqrt{\sum_{i=0}^{n-1}\int_{X_{i}}^{X_{i+1}}[L_{i}(S)-f(S)]^{2}g(S)\,dS}\leq \sqrt{2
  \sum_{i=0}^{n-1}h_{i}^{4}\int_{X_{i}}^{X_{i+1}}G(S)(f''(S))^2\,dS},
 \]
 where
 \begin{eqnarray*}
  G(S) &=& \widehat{G}\left(\frac{S-X_{i}}{h_{i}}\right),\\
   \widehat{G}(t) &\equiv& \int_{0}^{t}\widehat{g}_{i}(\xi)\frac{\xi^{2}(1-\xi)^{3}}{3}d\xi
   +\int_{t}^{1}\widehat{g}_{i}(\xi)\frac{(1-\xi)^{2}\xi^{3}}{3}d\xi,\\
   \widehat{g}_{i}(\xi) &\equiv& g(X_{i}+h_{i}\xi).
 \end{eqnarray*}
\end{theorem}
\begin{proof} See Appendix~\ref{appendix-1}.
\end{proof}

\section{Robust algorithms and convergence analysis}\label{robust-algorithm}

In this section we propose two
 algorithms  for static replication of nonlinear payoff functions with European call and put options. The first algorithm is to select the  optimal strike prices  when the strike prices of options from zero to infinity are all available. The second algorithm is to find the optimal weights of the options when there are only a limited number of strike prices are available.

We shall determine the values of strike prices $X_{i},\, i=1,\ldots,n-1$ (The boundary values $X_{0}$ and $X_{n}$ are fixed.) such that the error bound in Theorem~\ref{thm-error} is minimized, which can be achieved with the equidistribution equation\footnote{The equidistribution equation is studied in  \cite{Huang2005}.
The idea is simple and intuitive, for example, equidistributing the length of the function curve to be approximated results in that the nodes are clustered into the region where the function has large gradient and thus makes the approximation more accurate. In fact, the minimum expected area algorithm of \cite{Liu2010} is based on the equdistribution of the approximation error (no proof of this fact in his paper). Our robust algorithm is designed by equdistributing the upper bound on the approximation error which has explicit form as given by Theorem~\ref{thm-error}.}.

Following \cite{Huang2005}, we define the adaptation function $\rho_{i}$ and the
intensity parameter $\alpha_{h}$ by
\begin{eqnarray}\label{rho-i}
\rho_{i} &\equiv& \left(1+\frac{1}{\alpha_{h} h_{i}}\int_{X_{i}}^{X_{i+1}}
G(S)(f''(S))^{2} dS\right)^{\gamma/2},\\
\label{alpha-h}
 \alpha_{h} &\equiv& \left[\frac{1}{X_{n}-X_{0}}\sum_{i=0}^{n-1}h_{i}\left(
 \frac{1}{h_{i}}\int_{X_{i}}^{X_{i+1}}G(S)(f''(S))^2 dS\right)^{\gamma/2}\right]^{2/\gamma},
\end{eqnarray}
for some number\footnote{As mentioned by \cite{Huang2005}, the optimal value, which yields the smallest error bound, is $\gamma=2/5$.}  $\gamma\in (0,2]$.

The equidistribution equation for selecting strike prices $X_{1},\ldots,X_{n-1}$ is defined by
\begin{equation}\label{equidistribution-equation}
 h_{i}\rho_{i}=\frac{\sum_{j=0}^{n-1}h_{j}\rho_{j}}{n},\quad i=0,\ldots,n-1.
\end{equation}
Equation (\ref{equidistribution-equation}) can be written  equivalently as
\begin{equation}\label{equidistribution-equation-1}
 \sum_{\ell=0}^{i-1}h_{\ell}\rho_{\ell}=\frac{i}{n}
 \sum_{j=0}^{n-1}h_{j}\rho_{j},\quad i=1,\ldots,n.
\end{equation}
Define a piecewise constant function
\begin{equation*}\label{piecewise-constant-1}
 \overline{\rho}_{X}(x)=\rho_{i},\quad \hbox{when}\; x\in [X_{i},X_{i+1}],\; i=0,\ldots,n-1.
\end{equation*}
Then equation (\ref{equidistribution-equation-1}) can be rewritten as
\begin{equation}\label{equidistribution-equation-2}
 \int_{X_0}^{X_{i}}\overline{\rho}_{X}(x)\,dx =\frac{i}{n}\int_{X_0}^{X_{n}}\overline{\rho}_{X}(x)\,dx.
\end{equation}
Note that equation (\ref{equidistribution-equation-2})
cannot be solved exactly. We propose the following robust algorithm to solve the equidistribution equation.
\begin{algorithm}\label{algorithm}
 Set initial values
 \[
X_{i}^{(0)}=X_0+i\frac{X_{n}-X_{0}}{n}, \quad i=0,1,\ldots,n.
 \]
 Then the $(k+1)$th-step values for $k=0,1,\ldots$, are calculated by the following iteration
 \begin{equation}\label{iteration}
  \int_{X_{0}^{(k+1)}}^{X_{i}^{(k+1)}}\overline{\rho}_{X^{(k)}}(x)\,dx=\frac{i}{n}
  \int_{X_{0}^{(k)}}^{X_{n}^{(k)}}\overline{\rho}_{X^{(k)}}(x)\,dx,
 \end{equation}
 where $X_{0}^{(k+1)}\equiv X_{0},\, X_{n}^{(k+1)}\equiv X_{n}$ and $\overline{\rho}_{X^{(k)}}(x)$ is the piecewise constant function which is defined by
\[
 \overline{\rho}_{X^{(k)}}(x)=\rho_{i}^{(k)},\quad \hbox{when}\; x\in [X_{i}^{(k)},X_{i+1}^{(k)}],\; i=0,\ldots,n-1,
\]
where $\rho_{i}^{(k)}$ is the expression (\ref{rho-i}) with replacing $X_{i}$ by $X_{i}^{(k)}$.
\end{algorithm}

In fact the iteration equation (\ref{iteration}) explicitly determines
\begin{equation}\label{iteration-1}
 X_{i}^{(k+1)}=X_{j}^{(k)}+\frac{\frac{i}{n} \sum_{\ell=0}^{n-1}h_{\ell}^{(k)}\rho_{\ell}^{(k)}
 -\sum_{\ell=0}^{j-1}h_{\ell}^{(k)}\rho_{\ell}^{(k)}}
 {\rho_{j}^{(k)}},\quad i=1,\ldots,n-1,
\end{equation}
where $h_{\ell}^{(k)}\equiv X_{\ell+1}^{(k)}-X_{\ell}^{(k)}$ and the index $j$ is determined by
\begin{equation*}
\sum_{\ell=0}^{j-1}h_{\ell}^{(k)}\rho_{\ell}^{(k)}<\frac{i}{n} \sum_{\ell=0}^{n-1}h_{\ell}^{(k)}\rho_{\ell}^{(k)}\leq \sum_{\ell=0}^{j}h_{\ell}^{(k)}\rho_{\ell}^{(k)},
\end{equation*}
which means that
\[
 X_{j}^{(k)}< X_{i}^{(k+1)} \leq  X_{j+1}^{(k)}.
\]

In the implementation of Algorithm~\ref{algorithm}, we need to find $\rho_{j}^{(k)},\; j=0,\ldots,n-1$
in expression (\ref{iteration-1}). $\rho_{j}^{(k)}$ with $\gamma=2/5$ can be calculated approximately with some quadrature rules, e.g., the rectangle rule:
\[
 \rho_{j}^{(k)}\approx \left(1+\frac{G(X_{j+1}^{(k)})
 \left(f''(X_{j+1}^{(k)})\right)^2}{\left(\frac{1}{X_{n}-X_{0}}\sum_{\ell=0}^{n-1}
 h_{\ell}^{(k)}\left(G(X_{\ell+1}^{(k)})\right)^{1/5} \left(f''(X_{\ell+1}^{(k)})\right)^{2/5}\right)^{5}} \right)^{1/5},
\]
and
\begin{eqnarray*}
  G(X_{\ell+1}^{(k)}) &=& \widehat{G}(1)= \int_{0}^{1}\widehat{g}_{\ell}(\xi)\frac{\xi^{2}(1-\xi)^{3}}{3}d\xi,\\
  \widehat{g}_{\ell}(\xi) &\equiv& g(X^{(k)}_{\ell}+h^{(k)}_{\ell}\xi).
\end{eqnarray*}


The sequences $X_{i}^{(k)},\, i=0,1,\ldots,n$, generated by the iteration equation (\ref{iteration}) (or equivalent form (\ref{iteration-1})) converge to $X_{i},\, i=0,1,\ldots,n$, generated by the equidistribution equation (\ref{equidistribution-equation}) (or the equivalent forms (\ref{equidistribution-equation-1}), (\ref{equidistribution-equation-2})) as the iteration number $k\rightarrow +\infty$. Since the proof falls into the mathematical framework of \cite{Xu2011}, the details are omitted.
We only need to show the convergence rate of the approximation to the nonlinear payoff with the equidistribution equation (\ref{equidistribution-equation})  for selecting the strike prices.

\begin{theorem}\label{thm-convergence-rate}
 The convergence rate of the static replication of the nonlinear payoff $f$ using the linear spline approximation  (\ref{piecewise-linear})  with the equidistribution equation (\ref{equidistribution-equation})  for selecting the strike prices $X_{i},\; i=0,\ldots,n$, is given by
 \[
  \sqrt{\sum_{i=0}^{n-1}\int_{X_{i}}^{X_{i+1}}[L_{i}(S)-f(S)]^{2}g(S)\,dS}\leq C n^{-2},
 \]
 where $C$ is a positive constant that is independent of the strike prices $X_{i},\; i=1,\ldots,n-1$, and $n$ is the number of the strike prices used in the replication.
\end{theorem}
\begin{proof} See Appendix~\ref{appendix-2}.
\end{proof}


In the options market there are only limited number of  options with  fixed strike prices are traded. Suppose that the fixed strike prices are $\overline{X}_{j},\; j=1,\ldots,n$, in increasing order. We form a portfolio of call options at maturity to replicate the nonlinear payoff
 $f$,
\begin{equation*}
 f(S)\approx \Pi \equiv \sum_{j=1}^{n}w_{j}(S-\overline{X}_{j})^{+},
\end{equation*}
where $w_{i},\; i=1,\ldots,n$ are chosen to minimize the approximation error
\[
V(w_{1},\ldots,w_{n})\equiv \int_{0}^{\infty}[f(S)-\Pi]^{2}g(S)\,dS.
\]
The first-order optimality conditions lead to a system of equations
\begin{equation*}
 \mathbf{Q}\mathbf{w}=\mathbf{u},
\end{equation*}
with
$\mathbf{Q}=\left(q_{ij}\right)_{i,j=1,\ldots,n}$, $\mathbf{w}=(w_{1},\ldots,w_{n})^{T}$, $ \mathbf{u}=(u_{1},\ldots,u_{n})^{T}$
and
\begin{eqnarray}\label{q-ij}
q_{ij}&=&\int_{\max\{\overline{X}_{i},\overline{X}_{j}\}}^{\infty}
(S-\overline{X}_{i})(S-\overline{X}_{j})g(S)\,dS,\quad i,j=1,\ldots,n,\\
\label{u-i}
u_{i}&=&\int_{\overline{X}_{i}}^{\infty} (S-\overline{X}_{i})f(S)g(S)\,dS,\quad i=1,\ldots,n.
\end{eqnarray}
In general, for complex nonlinear payoff $f$, there is no explicit formula for $u_{i}$
and it relies on numerical solutions. Let $\widetilde{f}_i(S)\equiv (S-\overline{X}_{i})^{+}f(S)$. Then
\begin{equation*}
 u_{i}=\int_{0}^{\infty}\widetilde{f}_i(S)g(S)\,dS=E[\widetilde{f}(S_{T})],
\end{equation*}
from which we can apply Algorithm \ref{algorithm} to the new nonlinear payoff function $\widetilde{f}_i$ with $X_{0}\equiv \overline{X}_{i}$.

In many cases, $q_{ij}$ in (\ref{q-ij}) can be computed explicitly. Taking lognormal distribution for example (see models in Section~\ref{lognormal-case}), $q_{ij}$ can be calculated by the following formula (see \cite{Liu2010}):
\begin{equation*}
 q_{ij}=S_{0}^{2}\Phi(d_{0})e^{2r+\sigma^{2})T}-(\overline{X}_{i}
 +\overline{X}_{j})S_{0}\Phi(d_{1})e^{rT}+\overline{X}_{i}
 \overline{X}_{j}\Phi(d_{2}),
\end{equation*}
where $\Phi$ is the cumulative distribution function of a standard normal variable and
\[
d_{1}=\frac{\ln\left[ S_{0}/\max\{\overline{X}_{i},\overline{X}_{j}\}\right]
+(r+\sigma^{2}/2)T}{\sigma \sqrt{T}},\;
d_{2}=d_{1}-\sigma\sqrt{T},\; d_{0}=d_{1}+\sigma\sqrt{T}.
\]

\section{Numerical implementations and applications}

\subsection{Static replication under lognormal  process}\label{lognormal-case}

Let $S$ be a lognormal variable under a risk-neutral measure. Then $\ln S$ is a normal variable with mean $\ln S_0 +\left(r-\sigma^2/2\right)T$ and variance $\sigma^2 T$, where $r$ is the constant risk-free interest rate, $\sigma$ the constant volatility, and $T$ the maturity. The European call and put prices can be computed explicitly with the Black-Scholes formula. The data used in the numerical tests below are  $S_0=100$, $r=5\%$, $\sigma=20\%$ and  $T=0.25$.

\begin{example}
A variance swap:  Consider the following nonlinear payoff
\begin{equation}\label{example-payoff}
 f(S)=\frac{2}{T}\left(\frac{S-S_0}{S_0}-\ln \frac{S}{S_0}\right),
\end{equation}
which is studied in \cite{Liu2010} and \cite{Demeterfi1999}. This
payoff gives a \$1 exposure for one volatility point squared.

In Table~\ref{table-1} we list the replication values for different maturities and volatilities. The data used are $T=0.25,\, 0.5,\, 1$ and $\sigma=20\%,\, 30\%,\, 60\%$. We implement Algorithm~\ref{algorithm} to select $18$ strikes between $45$ and $140$ for volatility $\sigma=20\%$, $78$ strikes between $25$ and $200$ for volatility $\sigma=30\%$, and $158$ strikes between $15$ and $300$ for volatility $\sigma=60\%$.  The computational results are shown for a notional exposure of \$100 per volatility point squared. We see that the replication values using formula (\ref{replication-formula-Deme}) are very close to the true values of the nonlinear payoffs.

\begin{table}[htbp]
  \centering
  \caption{Replications by Algorithm~\ref{algorithm} (The numbers outside the brackets are the values of replications by Algorithm~\ref{algorithm} and those inside the brackets are the exact values of the nonlinear payoffs. The computational results are
shown for a notional exposure of \$100 per volatility point squared.)}\label{table-1}
\smallskip
\begin{tabular}{|c|c|c|c|}
\hline
\backslashbox{$T$}{$\sigma$} &  $20\%$  & $30\%$ &60\%\\
\hline
0.25 & 4.1122 (4.0123)    & 8.9664 (8.9502) & 35.6283 (35.6148) \\
\hline
0.5  & 4.0729 (4.0242)  & 8.9114 (8.9007) & 35.2220 (35.2341) \\
\hline
1 & 3.9718 (4.0467) & 8.7864 (8.8029) & 34.2173 (34.4861)\\
\hline
\end{tabular}
\end{table}

In Table~\ref{table-2}, we test the convergence rate of Algorithm~\ref{algorithm}. By increasing the total number of strikes between $45$ and $200$, we calculate the total value of the replication. We see that the replication values converge to the true value $(4.0123)$ of the given nonlinear payoff (\ref{example-payoff}) as the total number of strikes $n$ goes to infinity. In addition, we test the convergence rate as follows.
Let $\hbox{TRV}(n)$ denote the total replication values using $n$ points. Assume that the convergence rate of Algorithm~\ref{algorithm} is $p$, i.e.,
\begin{equation}\label{error-equ}
 |\hbox{Error}(n)|\equiv |\hbox{TRV}(n)-(\hbox{True Value})|=O(n^{-p}),
\end{equation}
where $O(n^{-p})$ means that there exists a positive constant $C$ such that
$  O(n^{-p})\approx Cn^{-p}$.
Equation (\ref{error-equ}) gives the formula
for testing the convergence rate
\begin{equation}\label{rate-test-formula}
 p\approx\frac{\log\left(|\hbox{Error}(n)|/|
 \hbox{Error}(2n)|\right)}{\log 2}.
\end{equation}
In Table~\ref{table-2} we calculate the value $p$ using formula (\ref{rate-test-formula})
and obtain that $p\approx 2$. It means that the convergence rate of Algorithm~\ref{algorithm} is $2$, which is consistent with the theoretical result of Theorem \ref{thm-convergence-rate}.

\begin{table}[htbp]
  \centering
  \caption{Convergence rates of Algorithm~\ref{algorithm}
  for replications (The computational results are
shown for a notional exposure of \$100 per volatility point squared.)}\label{table-2}
\smallskip
\begin{tabular}{cccc}
 \hline
 Number of strikes ($n$)&	Total replication values (TRV) & $\hbox{Error}(n)$ &	Convergence rates ($p$)\\
20&	4.1651&	0.1528&--\\
40&	4.0484&	0.0361& 2.1\\
80&	4.0211&	0.0088& 2.0\\
160&	4.0145&	0.0022& 2.0\\
320&	4.0128&	0.0005& 2.1\\
640& 4.0124 & 0.0001&2.3\\
 \hline
\end{tabular}
\end{table}

\smallskip

In Table~\ref{quadratic-replication-test} we test for the static quadratic replication algorithm for nonlinear payoff (\ref{example-payoff}). A set of strikes
$\{50,\, 70,\, 90,\, 100,\, 110,\, 130\}$ are used for the replications\footnote{In the real applications, the strike prices can be picked up from the traded option market.}. The optimal weights and the replication values are computed by the static quadratic replication algorithms in Section~\ref{robust-algorithm}. The numerical results in Table~\ref{quadratic-replication-test} show that the total replication value is $4.0224$ which is close to the true value $4.0123$.

\begin{table}[htbp]
  \centering
  \caption{Numerical results for the static quadratic replication algorithms (The computational results are
shown for a notional exposure of \$100 per volatility point squared.)}\label{quadratic-replication-test}
\smallskip
\begin{tabular}{cccc}
 \hline
 Strikes &	Weight & Value per option  & Cost today\\
50&	1.7393&	  50.6211& 88.0450 \\
70&	$-3.3196$& 30.8698& $-102.4741$\\
90&	1.2107&	  11.6701 & 14.1288\\
100&	0.7073&	4.6150& 3.2642\\
110&	0.8639&	 1.1911&1.0290 \\
130&   1.2978&  0.0228 & 0.0296 \\
\hline
Total & \multicolumn{2}{c}{}& 4.0224 \\
 \hline
\end{tabular}
\end{table}

\end{example}

\begin{example}
 A swaption: Consider the following nonlinear payoff
 \begin{equation}\label{payoff-call-swaption}
  f^{c}(S)=\left( \frac{2}{T}\left(\frac{S-S_0}{S_0}-\ln \frac{S}{S_0}\right)-K\right)^{+}.
 \end{equation}
 We compute the replication value of nonlinear payoff (\ref{payoff-call-swaption}) using Algorithm~\ref{algorithm}. It is  more convenient to replicate the put swaption with
\begin{equation}\label{payoff-put-swaption}
  f^{p}(S)=\left(K- \frac{2}{T}\left(\frac{S-S_0}{S_0}-\ln \frac{S}{S_0}\right)\right)^{+},
 \end{equation}
 and then the call swaption can be computed easily by the put-call parity relation.
 Let
 \begin{equation*}
  h(S)=K- \frac{2}{T}\left(\frac{S-S_0}{S_0}-\ln \frac{S}{S_0}\right),\quad S> 0.
 \end{equation*}


A simple check shows that $h$   is strictly concave,  has the maximum value
 at $S=S_{0}$ and has only two solutions  $S_{L}$ and $S_{R}$ ($S_{L}<S_{R}$) to the nonlinear equation
$ h(S)=0$.

The strikes used in the replication can be selected between $S_{L}$ and $S_{R}$. The values of $S_{L}$ and $S_{R}$ are calculated by Newton's method:
\begin{equation}\label{Newton}
 S^{(k)}=S^{(k-1)}-\frac{h\left(S^{(k-1)}\right)}{h'\left(S^{(k-1)}\right)},\quad k=1,\ldots
\end{equation}
with the initial point $S^{(0)}$ to be chosen sufficiently small (and large) such that $h(S^{(0)})<0$. Then Newton's iteration (\ref{Newton}) converges to $S_{L}$ (and $S_{R}$) quadratically.

In Table~\ref{table-swaption} we list the replication values of the call swaption for different maturities and volatilities with  $S_0=100$, $r=5\%$ and $K=0.01$. Using Newton's iteration (\ref{Newton}), we obtain the values
$S_{L}\approx 95.0840; S_{R}\approx 105.0827$ for $T=0.25$, $S_{L}\approx 93.0956; S_{R}\approx 107.2377$ for $T=0.5$, $S_{L}\approx 90.3315; S_{R}\approx 110.3351$ for $T=1$. $18$ strikes between $S_L$ and $S_{R}$ are selected using Algorithm~\ref{algorithm}.

\begin{table}[htbp]
  \centering
  \caption{Replications by Algorithm~\ref{algorithm} for swaption (\ref{payoff-call-swaption})  (The numbers outside the brackets are the values of replications by Algorithm~\ref{algorithm} and those inside the brackets are the values of Monte-Carlo simulation. The computational results are
shown for a notional exposure of \$100 per volatility point squared.)}\label{table-swaption}
\smallskip
\begin{tabular}{|c|c|c|c|}
\hline
\backslashbox{$T$}{$\sigma$} &  $20\%$  & $30\%$ &60\%\\
\hline
0.25 & 3.2796 (3.2791)    & 8.1353 (8.1469) &  34.7138 (34.6813) \\
\hline
0.5  & 3.2998 (3.3019) &  8.0960 (8.0907)&  34.3438 (34.3599)\\
\hline
1 & 3.3389 (3.3457) &   8.0180 (8.0034)& 33.6168 (33.5928)\\
\hline
\end{tabular}
\end{table}

\end{example}

\smallskip

The replication methods can also be applied to nonlinear path-dependent payoffs. Consider, for example, $f(M_{T})$ with $M_{T}={\displaystyle \max_{0\leq t\leq T} S_{t}}$.
It is known from \cite{Shreve2004} that the pdf of $M_{T}$ has an explicit form. Therefore, Algorithm~\ref{algorithm} can be used and the value of $f(M_{T})$ can be
replicated by a portfolio of barrier options with payoff $\textbf{1}_{M_{T}\geq K}$
and lookback options with payoffs $\left(M_{T}-K\right)^{+}$ or $\left(K-M_{T}\right)^{+}$.

\subsection{Static replication under counterparty risk}

In this section we consider a financial market model with a risky asset subject to a counterparty risk: the dynamics of the risky asset is affected
by the counterparty which may default. However, this stock still exists and can be traded after the default of the counterparty.

Let $W=(W_{t})_{t\in [0,T]}$ be a Brownian motion over a finite horizon $T<\infty$
with the probability space $(\Omega,\mathcal{G},P)$ and denote by $\mathbb{F}= (\mathcal{F}_t)_{t\in [0,T]}$ the natural filtration of $W$. Let $\tau$, a nonnegative and finite random variable on $(\Omega,\mathcal{G},P)$, represent the default time
before the default time $\tau$, the filtration $\mathbb{F}$ represents the information accessible to the investors. When the default occurs, the investors add this new information $\tau$ to the reference filtration $\mathbb{F}$.

Write the risky asset price $S_{t}$ into the following form
\[
 S_{t}=S_{t}^{\mathbb{F}} \textbf{1}_{t<\tau}+S_{t}^{d}(\tau)\textbf{1}_{t\geq \tau},
 \quad 0\leq t\leq T,
\]
where $S_{t}^{\mathbb{F}}$ is $\mathbb{F}$-adapted and $S_{t}^{d}(\theta)$ is $\theta$-measurable and $\mathbb{F}$-adapted.
Then we assume that the asset price follows the following dynamics under physical measure:
\begin{eqnarray}\label{dynam-phy-meas-a}
 &&d S_{t}^{\mathbb{F}}= S_{t}^{\mathbb{F}}\left(\mu^{\mathbb{F}}dt +\sigma^{\mathbb{F}}dW_{t}\right),\quad 0\leq t<\tau,\\
 \label{dynam-phy-meas-b}
 &&d S_{t}^{d}(\tau) = S_{t}^{d}(\tau)\left(\mu_{t}^{d}(\tau)dt +\sigma_{t}^{d}(\tau)dW_{t}\right),\quad
 \tau<t\leq T,\\
 \label{dynam-phy-meas-c}
 &&S_{\tau}^{d}(\tau)= S^{\mathbb{F}}_{\tau-}(1-\gamma^{\mathbb{F}}_{\tau}).
\end{eqnarray}
Here for simplicity we assume that
\[
\mu^{\mathbb{F}}=\mu_1,\; \sigma^{\mathbb{F}}=\sigma_{1},\; \mu_{t}^{d}(\tau)=\mu_{2},\; \sigma_{t}^{d}(\tau)=\sigma_{2},\;\gamma^{\mathbb{F}}_{\tau}=\gamma,
\]
where $\mu_{1},\, \sigma_{1},\, \mu_{2},\, \sigma_{2}$ are nonnegative constants and $\gamma\; (\gamma\leq 1)$  satisfies a fixed distribution. Moreover $\gamma,\, \tau,\, W_{t}$ are independent and $\tau$ is an exponential variable with parameter $\lambda$. For more general set-ups on the model, the reader is  referred to \cite{Jiao-Pham2011}.

Assume that $r$ is riskless interest rate. Changing measure with the Girsanov theorem, the dynamics (\ref{dynam-phy-meas-a})--(\ref{dynam-phy-meas-c}) for asset price $S_{t}$ under physical measure are transformed into the following form under equivalent martingale measure
\begin{eqnarray}\label{dynam-martingale-meas-a}
 &&d S_{t}^{\mathbb{F}}= S_{t}^{\mathbb{F}}\left((r+\lambda m)dt +\sigma_{1}dW_{t}\right),\quad 0\leq t<\tau,\\
 \label{dynam-martingale-meas-b}
 &&d S_{t}^{d}(\tau) = S_{t}^{d}(\tau)\left(rdt +\sigma_{2}dW_{t}\right),\quad
 \tau<t\leq T,\\
 \label{dynam-martingale-meas-c}
 &&S_{\tau}^{d}(\tau)= S^{\mathbb{F}}_{\tau-}(1-\gamma),
\end{eqnarray}
where $m=E(\gamma)$.
From (\ref{dynam-martingale-meas-a})--(\ref{dynam-martingale-meas-c}), we can see that if $\gamma=0$ then there is no jump of asset price at time $\tau$ and this is a simple regime switching model. In practice, we may assume $\gamma$ is a discrete random variable to simplify the computation, e.g., we may assume that $\gamma$ takes value $\gamma_i$ with probability $p_i$ for $i=1,2,3$, where  $0<\gamma_{1}\leq 1$ (loss), $\gamma_{2}=0$ (no change) and  $\gamma_{3}<0$ (gain).
The distribution function of random variable $S=S_{T}$ is given by
\begin{eqnarray}\label{density-regime-diffusion}
 F(S)&=&e^{-\lambda T}\Phi\left(\frac{\ln(S/S_0) -a(T)}{b(T)}\right)\\
 &&\quad+\,\sum_{i=1}^{3}p_{i}\int_{0}^{T}\lambda e^{-\lambda t}\Phi\left(\frac{1}{b(t)}\left(\ln\left(\frac{S}{S_{0}(1-\gamma_{i})}\right) -a(t)\right)\right)\,dt,\nonumber
\end{eqnarray}
where $a(t)=(r+\lambda m-\sigma_{1}^{2}/2)t+(r-\sigma_{2}^{2}/2)(T-t)$, $b(t)=\sqrt{\sigma_{1}^{2} t+\sigma_{2}^{2}(T-t)}$ and $\Phi$ is the cumulative distribution function of a standard normal variable. The proof of formula (\ref{density-regime-diffusion}) is given in Appendix~\ref{appendix-3}.

Combining the distribution function $F$ in (\ref{density-regime-diffusion}) and  the formula
$$
\int_{0}^{\infty}(S-K)^{+}d\Phi\left(\frac{1}{B}\left(\ln\left(\frac{S}{C}\right)-A\right)\right)
=Ce^{A+\frac{B^2}{2}}\Phi(x_{0}+B)-K\Phi(x_{0}),
$$
where $A$ is a constant, $B,C, K$ are positive constants and $x_{0}=\frac{1}{B}\left(A-\ln\left(\frac{K}{C}\right)\right)$, we can easily compute the value of a call option at time 0 with counterparty risk as
\begin{eqnarray}\label{call-counterparty}
&&e^{-rT}E\left[(S-K)^{+}\right]\\
 &&\quad = S_{0}e^{-(1-m)\lambda T}
\Phi\left(\widetilde{d}_{0}+b(T)\right)
-Ke^{-(r+\lambda)T}
\Phi\left(\widetilde{d}_{0}\right)\nonumber\\
&&\quad+\,e^{-rT}\sum_{i=1}^{3}p_{i}\int_{0}^{T}\lambda e^{-\lambda t}\left[
S_{0}(1-\gamma_i)e^{a(t)+b^{2}(t)/2}\Phi\left(\widetilde{d}_{i}(t)+b(t)\right)
-K\Phi\left(\widetilde{d}_{i}(t)\right)
\right]\,dt,\nonumber
\end{eqnarray}
where $\widetilde{d}_{0}={1\over b(T)} \left(a(T)-\ln\left({K\over S_{0}}\right)\right)$ and
$\widetilde{d}_{i}(t)= \frac{1}{b(t)}\left(a(t)-\ln\left(\frac{K}{S_0(1-\gamma_{i})}\right)\right)$ for $i=1,2,3$.
The value of a put option  can be computed with the put-call parity relation
\begin{eqnarray}\label{put-counterparty}
&&e^{-rT}E\left[(K-S)^{+}\right]-e^{-rT}E\left[(S-K)^{+}\right]\\
 &&\quad = Ke^{-rT}-S_{0}e^{-(1-m)\lambda T} -
\lambda S_{0}e^{-rT}\left(\int_{0}^{T}
e^{a(t)+b^{2}(t)/2-\lambda t}
\,dt\right)\sum_{i=1}^{3}p_{i}(1-\gamma_i).\nonumber
\end{eqnarray}

Now we can use Algorithm~\ref{algorithm} to replicate a variance swap with the payoff $f$ in (\ref{example-payoff}) and
with the asset price $S$ having the distribution function $F$ in (\ref{density-regime-diffusion}). In Table \ref{table-for-counterparty-risk} we list the replication values and the true values    for different jump sizes and probabilities. The data used are  $S_0=100$, $T=1$, $r=5\%$, $\sigma_1=40\%$, $\sigma_{2}=20\%$, $\lambda=0.5$, others are in Table~\ref{table-for-counterparty-risk}. Since $f$  is simple we have  a closed-form formula for the valuation that can be used to  compare the accuracy of the algorithm. We  choose $X_0=5$ and $X_n=400$ and $n=80$. It is clear that the static replication values are very close to the true values.

\begin{table}[htbp]
  \centering
  \caption{Numerical results for the static replication of nonlinear payoff under counterparty risk (The computational results are
shown for a notional exposure of \$100 per volatility point squared.)}\label{table-for-counterparty-risk}
\smallskip
\begin{tabular}{cccccccc}
 \hline
 $\gamma_1$ & $\gamma_2$ & $\gamma_3$ & $p_1$ & $p_2$ & $p_3$ & Replication value& True value\\
 \hline
 0.5& 0 & $-0.2$ & 0.3& 0.5& 0.2 & 17.6584 & 17.6316 \\
 \hline
 0.9& 0 & $-0.2$ & 1 & 0& 0& 118.0538 & 118.0215\\
 \hline
 0.9&	0&	$-0.2$ & 0.9&	0&	0.1& 107.6932 & 107.6547 \\
 \hline
\end{tabular}
\end{table}

\section{Conclusions}

In this paper we  propose a robust algorithm for optimally approximating the nonlinear payoff and derive the rigorous convergence theory. We define an equidistribution equation for selecting the strike prices and construct a simple, fast and accurate iterative algorithm for  implementation.  In addition we perform some numerical tests, including examples of  the static quadratic replication with the options traded in the market and the asset price model with  counterparty risk. The results of the paper  have  generalized and improved those of the static replication and approximation in  the literature.

\appendix

\section{Proof of Theorem~\ref{thm-error}}\label{appendix-1}

Transform $S\in [X_{i},X_{i+1}]$ into $\xi\in[0,1]$ by mapping
$ S=X_{i}+h_{i}\xi$ and denote by
 $\widehat{f}_{i}(\xi) \equiv f(X_{i}+h_{i}\xi)$. Then,
from (\ref{piecewise-linear}), we have
\begin{equation}\label{transformed-piecewise-linear}
 \widehat{L}_{i}(\xi)=L_{i}(X_{i}+h_{i}\xi)=\widehat{f}_{i}(0)(1-\xi) +\widehat{f}_{i}(1)\xi,\quad \xi\in[0,1].
\end{equation}
Taylor's theorem gives that
\begin{eqnarray}\label{taylor-1}
\widehat{f}_{i}(0)&=&\widehat{f}_{i}(\xi) -\xi\widehat{f}'_{i}(\xi)-\int_{\xi}^{0}t\widehat{f}''_{i}(t)dt,\\
\label{taylor-2}
\widehat{f}_{i}(1)&=&\widehat{f}_{i}(\xi)+(1-\xi)\widehat{f}'_{i}(\xi)
+\int_{\xi}^{1}(1-t)\widehat{f}''_{i}(t)dt.
\end{eqnarray}
Using (\ref{transformed-piecewise-linear}), (\ref{taylor-1}) and (\ref{taylor-2}), we have
\[
\widehat{f}_{i}(\xi)-\widehat{L}_{i}(\xi)
=-\xi\int_{\xi}^{1}(1-t)\widehat{f}''_{i}(t)dt -(1-\xi)\int_{0}^{\xi}t\widehat{f}''_{i}(t)dt.
\]
Therefore, using $(a+b)^2\leq 2(a^2+b^2)$ and the Cauchy-Schwartz inequality, we derive that
\begin{eqnarray}\label{error-estimate-1}
&&\int_{0}^{1}[\widehat{f}_{i}(\xi)-\widehat{L}_{i}(\xi)]^2\widehat{g}(\xi)d\xi\\
&&\leq\, 2\int_{0}^{1}\left(\xi^2\left(\int_{\xi}^{1}(1-t)\widehat{f}''_{i}(t)dt\right)^2
+ (1-\xi)^2\left(\int_{0}^{\xi}t\widehat{f}''_{i}(t)dt\right)^2\right)
\widehat{g}_{i}(\xi)d\xi\nonumber\\
&&\leq\, 2\int_{0}^{1}\left(\xi^{2}{(1-\xi)^3\over 3}
\int_{\xi}^{1}(\widehat{f}''_{i}(t))^2dt+
(1-\xi)^{2}{\xi^3\over 3}
\int_{0}^{\xi}(\widehat{f}''_{i}(t))^2dt\right)\,\widehat{g}_{i}(\xi)d\xi\nonumber\\&&=\, 2\int_{0}^{1}\widehat{G}(t)(\widehat{f}''_{i}(t))^2dt,\nonumber
\end{eqnarray}
where
\begin{equation}\label{G-hat}
\widehat{G}(t)\equiv \int_{0}^{t}\widehat{g}_{i}(\xi)\frac{\xi^{2}(1-\xi)^{3}}{3}d\xi
+\int_{t}^{1}\widehat{g}_{i}(\xi)\frac{(1-\xi)^{2}\xi^{3}}{3}d\xi.
\end{equation}
We can now estimate the weighted squared error
\begin{eqnarray}\label{error-estimate-3}
 \int_{X_{i}}^{X_{i+1}}[L_{i}(S)-f(S)]^{2}g(S)\,dS
&=& h_{i}\int_{0}^{1}[\widehat{f}_{i}(\xi)-\widehat{L}_{i}(\xi)]^2\widehat{g}(\xi)d\xi\nonumber\\
 &\leq& 2h_{i}\int_{0}^{1}\widehat{G}(\xi)\left[\frac{d^2f(X_{i}+h_{i}\xi)}
 {d\xi^2}\right]^2d\xi\nonumber\\
 &=& 2h_{i}^{4}\int_{X_{i}}^{X_{i+1}}G(S)
 \left(f''(S)\right)^2 dS,
\end{eqnarray}
where
\begin{equation*}
G(S)\equiv \widehat{G}\left(\frac{S-X_{i}}{h_{i}}\right),\quad \hbox{for}\; S\in [X_{i},X_{i+1}].
\end{equation*}
\hfill $\Box$

\section{Proof of Theorem~\ref{thm-convergence-rate}}\label{appendix-2}

Following the idea of \cite{Huang2005} using a different measure, we prove this theorem. First we prove that $\sum_{j=0}^{n-1}h_{j}\rho_{j}$ is bounded. It follows from Jensen's inequality and the definition (\ref{alpha-h}) for $\alpha_{h}$ that
\begin{eqnarray}\label{thm2-ineq-1}
 \sum_{j=0}^{n-1}h_{j}\rho_{j}&=& \sum_{j=0}^{n-1}h_{j}\left(1+\alpha_{h}^{-1}\left(
 \frac{1}{h_{j}}\int_{X_{j}}^{X_{j+1}} G(S)(f''(S))^{2}dS\right)\right)^{\gamma/2}\nonumber\\
 &\leq& \sum_{j=0}^{n-1}h_{j}\left(1+\alpha_{h}^{-\gamma/2}\left(
 \frac{1}{h_{j}}\int_{X_{j}}^{X_{j+1}} G(S)(f''(S))^{2}dS\right)^{\gamma/2}\right)\nonumber\\
 &=& 2(X_{n}-X_{0}).
\end{eqnarray}
From the error bound in Theorem~\ref{thm-error}, the definition (\ref{alpha-h}) for $\alpha_{h}$, and the definition (\ref{rho-i}) for $\rho_{i}$, we derive that
\begin{eqnarray}\label{thm2-ineq-2}
&&\sqrt{\sum_{i=0}^{n-1}\int_{X_{i}}^{X_{i+1}}[L_{i}(S)-f(S)]^{2}g(S)\,dS}\nonumber\\
&&\quad\leq\, \sqrt{2
  \sum_{i=0}^{n-1}h_{i}^{5}\left(\alpha_{h}+\frac{1}{h_{i}}
  \int_{X_{i}}^{X_{i+1}}G(S)(f''(S))^2\,dS\right)}\nonumber\\
&&\quad =\, \sqrt{2\alpha_{h}
  \sum_{i=0}^{n-1}h_{i}^{5}\rho_{i}^{2/\gamma}}.
\end{eqnarray}
Now take $\gamma=2/5$. Then combining (\ref{thm2-ineq-1}) with (\ref{thm2-ineq-2}) and using the definition of (\ref{equidistribution-equation}),
we derive that
\begin{eqnarray*}
\sqrt{\sum_{i=0}^{n-1}\int_{X_{i}}^{X_{i+1}}[L_{i}(S)-f(S)]^{2}g(S)\,dS}
&\leq& \sqrt{2\alpha_{h}\sum_{i=0}^{n-1}h_{i}\rho_{i}\left(h_{i}\rho_{i}\right)^{4}}\\
&=& \sqrt{2\alpha_{h}\sum_{i=0}^{n-1}h_{i}\rho_{i}\left(\frac{\sum_{j=0}^{n-1} h_{j}\rho_{j}}{n}\right)^{4}}\\
&\leq& \sqrt{2^{6}\alpha_{h}\left(X_{n}-X_{0}\right)^{5}}n^{-2}\\
&\leq& Cn^{-2}.
\end{eqnarray*}
Here we have used the fact that $\alpha_{h}$ is bounded (see  \cite{Huang2005}) in the last inequality.  \hfill $\Box$

\section{Proof of formula (\ref{density-regime-diffusion})}\label{appendix-3}

The solutions to (\ref{dynam-martingale-meas-a})--(\ref{dynam-martingale-meas-c}) are given by
\begin{eqnarray}
 S_{t}^{\mathbb{F}}&=&S_{0}e^{(r+\lambda m -\sigma_{1}^{2}/2)t+\sigma_{1}W_{t}},\quad
 0\leq t <\tau, \label{C1}\\
 S_{t}^{d}(\tau)&=&S_{\tau}^{d}(\tau)e^{(r-\sigma_{2}^{2}/2)(t-\tau) +\sigma_{2}(W_{t}-W_{\tau})},\quad \tau<t\leq T, \label{C2}\\
 S_{\tau}^{d}(\tau)&=& S_{\tau-}^{\mathbb{F}}(1-\gamma). \label{C3}
\end{eqnarray}
We now compute the distribution of $S_{T}$.
\begin{eqnarray}\label{distribution-fun}
 F(S)&=&P(S_{T}\leq S)=E[\textbf{1}_{S_{T}\leq S}]\nonumber\\
 &=&E[\textbf{1}_{S_{T}\leq S}\textbf{1}_{\tau\geq T}]+E[\textbf{1}_{S_{T}\leq S}
 \textbf{1}_{\tau<T}]\nonumber\\
 &\equiv& A+B.
\end{eqnarray}
$A$ and $B$ can be computed as follows:
\begin{eqnarray}\label{A}
 A&=&E\left[E[\textbf{1}_{S_{T}\leq S}\textbf{1}_{\tau\geq T}\big|\tau]\right]\nonumber\\
 &=& \int_{T}^{\infty}\lambda e^{-\lambda t}
 E[\textbf{1}_{S_{T}^{\mathbb{F}}\leq S}\big|\tau=t]\,dt\nonumber\\
 &=&\int_{T}^{\infty}\lambda e^{-\lambda t}P(S_{T}^{\mathbb{F}}\leq S)\,dt\nonumber\\
 &=&e^{-\lambda T}\Phi\left(
 \frac{\ln\left(S/S_{0}\right)-(r+\lambda m-\sigma_{1}^{2}/2)T}
 {\sigma_{1}\sqrt{T}}\right)
 \end{eqnarray}
and
\begin{eqnarray}\label{B-1}
B&=& E\left[E[\textbf{1}_{S_{T}\leq S}\textbf{1}_{\tau<T}\big|\tau]\right]\nonumber\\
&=& \int_{0}^{T}\lambda e^{-\lambda t}E[\textbf{1}_{S_{T}^{d}(\tau)\leq S}\big|\tau=t]\, dt\nonumber\\
&=& \int_{0}^{T}\lambda e^{-\lambda t}
 E\left[P(S_{T}^{d}(t)\leq S\big| \gamma)\right]\, dt\nonumber\\
&=& \int_{0}^{T}\lambda e^{-\lambda t}
 E\left[\Phi
 \left(\frac{1}{b(t)}\left(\ln\left(\frac{S}{S_{0}(1-\gamma)}\right)-a(t)\right)\right) \right]\, dt\nonumber\\
&=& \sum_{i=1}^{3}p_{i}\int_{0}^{T}\lambda e^{-\lambda t}\Phi
 \left(\frac{1}{b(t)}\left(\ln\left(\frac{S}{S_{0}(1-\gamma_{i})}\right)
 -a(t)\right)\right)\,dt,
\end{eqnarray}
where
$a(t)\equiv (r+\lambda m-\sigma_{1}^{2}/2)t+(r-\sigma_{2}^{2}/2)(T-t)$
and $
b(t)= \sqrt{\sigma_{1}^{2}t+\sigma_{2}^{2}(T-t)}$. We have used (\ref{C1})-(\ref{C3}) in computing $P(S_{T}^{d}(t)\leq S\big| \gamma)$.
\hfill $\Box$
\end{document}